\documentclass[letterpaper, 10 pt, conference]{ieeeconf}
\IEEEoverridecommandlockouts

\let\proof\relax
\let\endproof\relax
\usepackage{amsthm}
\usepackage[utf8]{inputenc}
\usepackage{amsmath,amssymb}
\usepackage{booktabs}
\usepackage{graphicx}
\usepackage[font=small]{subcaption}
\usepackage[dvipsnames]{xcolor}
\usepackage{numprint}
\usepackage{csvsimple}
\usepackage{comment}
\usepackage{amsfonts}
\usepackage{xcolor}
\usepackage{subcaption}
\usepackage{tikz}
\usepackage[linesnumbered,ruled,noend]{algorithm2e}
\usetikzlibrary{automata, positioning, arrows}
\tikzset{initial text={}} 

\usepackage{tabularx,siunitx}
\usepackage{cite}
\usepackage{rotating} 
\usepackage{hyperref}
\usepackage[pscoord]{eso-pic} 

\newtheorem{theorem}{Theorem}
\newtheorem{problem}{Problem}

\newtheorem{proposition}{Proposition}

\newtheorem{lemma}{Lemma}
\newtheorem{definition}{Definition}

\newtheorem{example}{Example}

\DeclareMathOperator*{\realnumbers}{\mathbb{R}}

\newcommand{\predicate}{\gamma}
\newcommand{\Predicate}{\Gamma}
\newcommand{\pu}{\mathbf{u}}
\newcommand{\px}{\mathbf{x}}

\newcommand{\la}{\langle}
\newcommand{\ra}{\rangle}

\newcommand{\tempop}[1]{\mathcal{#1}}
\newcommand{\until}{\,\tempop{U}}
\newcommand{\eventually}{\,\Diamond} 
\newcommand{\globally}{\,\Box} 

\newcommand{\stlnn}{STL$_{nn}$\xspace}
\newcommand{\timedomain}{\mathbb{T}}
\newcommand{\TA}{\mathcal{T} \hspace{-1mm} \mathcal{A}}

\newcommand{\placetextbox}[3]{
  \setbox0=\hbox{#3}
  \AddToShipoutPictureFG*{
    \put(\LenToUnit{#1\paperwidth},\LenToUnit{#2\paperheight}){\vtop{{\null}\makebox[0pt][c]{#3}}}%
  }%
}%

\title{\LARGE \bf Automaton-Guided Control Synthesis for Signal Temporal \\Logic Specifications}
\author{Qi Heng Ho, Roland B. Ilyes, Zachary N. Sunberg, and Morteza Lahijanian
\thanks{Authors are with the department of Aerospace Engineering Sciences at the University of Colorado Boulder, CO, USA
        {\tt\small \{\textit{firstname}.\textit{lastname}\}@colorado.edu}}%
}

\begin{document}
\placetextbox{0.5}{0.95}{To appear in the 61st IEEE Conference on Decision and Control (CDC), December 2022.}
\maketitle

\begin{abstract}
    This paper presents an algorithmic framework for control synthesis for continuous dynamical systems subject to 
    signal temporal logic (STL) specifications. 
    We propose a novel algorithm to obtain a time-partitioned finite automaton from an STL specification, and introduce a multi-layered framework that utilizes this automaton to guide a sampling-based search tree both spatially and temporally. Our approach is able to synthesize a controller for nonlinear dynamics and polynomial predicate functions.  We prove the correctness and probabilistic completeness of our algorithm, and illustrate the efficacy and efficiency of our framework on several case studies.  Our results show an order of magnitude speedup over the state of the art.
\end{abstract}
\section{Introduction}
Control synthesis for complex behaviors is an essential challenge for autonomous systems. Temporal Logics (TL)
admit expressive formalisms to describe complex temporal properties of control systems. Namely, Linear TL (LTL) \cite{Baier2008}, which can express order of events over discrete time, has been well-studied for control synthesis. LTL algorithms are mostly automata-based and perform finite abstraction of continuous-space systems, e.g., \cite{Lahijanian2009, belta2017formal, Bhatia2010, Maly2013, kress2018synthesis}. For continuous-time systems, more appropriate languages are Metric Interval TL (MITL) and Signal TL (STL) 
since they allow the expression of real-time temporal properties \cite{Maler2004}. However, the additional expressivity of these logics comes with additional computational complexity, leading to scalability challenges, especially for systems with complex dynamics. In this paper, we aim to develop an efficient STL control synthesis framework for systems with general (possibly nonlinear) dynamics.


The major difference between STL and MITL lies in that STL is interpreted over real-valued signals whereas MITL is interpreted over boolean signals. Work \cite{Alur1996} shows that MITL formulas can be translated to timed automata with multiple clocks, which results in additional complexity for synthesis. For STL, on the other hand, 
there does not exist any algorithm to translate its formulas to automata.  
The difference between the two logics practically dissolves when the signals are well-behaving \cite{Maler2004}, which is generally the case for control synthesis.  That is, there exists a direct mapping between real-valued and boolean signals. However, this mapping is difficult to compute algorithmically if the predicates of the signals are defined over nonlinear functions.

In the controls community, optimization-based methods have been a popular approach to STL control synthesis. Work \cite{raman2014} proposed a Mixed Integer Linear Program (MILP) encoding for model predictive control constraints. MILP is known to be NP-hard, which may be computationally inefficient, and with no guarantees on constraint satisfaction. Another work uses control barrier functions on a fragment of STL for control over short time horizons \cite{Lindemann2019}. These works, however, are constrained to linear predicates \cite{raman2014} or control-affine systems \cite{Lindemann2019}.

Recently, sampling-based methods have been employed to solve real-time TL tasks.  Work \cite{stl-rrtstar} uses an RRT* approach with time sampling in order to grow a tree that maximally satisfies an STL specification. They assume linear predicates, and the availability of a steer function that precisely steers the system from a state to another state given an exact time, which is hard to design for general dynamical systems. Work \cite{Barbosa2019gae} proposes a cost function to guide exploration to satisfy a restricted safety fragment of STL.
Work \cite{Barbosa2019} introduces a sampling-based MITL synthesis.  It first translates an MITL formula to a timed automaton, which is then converted to a time-abstract zone automaton. Then, it uses a geometric RRT* approach to calculate untimed waypoints and set up a linear program to compute time stamps in order to generate a timed path. Thereafter, a control barrier function is used to ensure dynamics feasibility. 
That work however is limited to propositional regions in the form of convex polytopes and control affine systems.

In this paper, we introduce an efficient sampling-based STL automaton-guided control synthesis framework that is able to handle both nonlinear dynamics with input constraints and nonlinear predicates. Specifically, we focus on a fragment of STL with \textit{non-nested temporal operators} (\stlnn), which maintains all the key operators of STL. This fragment subsumes the fragment considered in \cite{Lindemann2019} by including negation and boolean combination of temporal operators, and is expressive for a wide range of timed temporal tasks. We present a novel algorithm to efficiently translate \stlnn formulas into finite automata that can be used to guide a sampling-based search, both spatially and temporally. We then present a construction of state space abstraction based on nonlinear predicates, which we use in conjunction with the automaton to produce high-level discrete guides in synergy with low-level continuous planning. We prove that this framework is correct-by-construction and probabilitically complete. Our case studies illustrate that the framework is not only able to handle general nonlinear dynamics and polynomial predicate functions, but also that it is at least an order of magnitude faster than the state of the art.


In summary, the contributions of the paper are four-fold: (1) a novel method for efficient conversion of \stlnn formulas to finite automata that precisely accept the language of the formulas based on techniques from untimed TL, (2) a use of syntactic separation to partition time in order to improve the granularity of the abstraction in the time dimension for efficient synthesis, (3) a multi-layered synthesis framework that utilizes the constructed automaton to guide a sampling-based search both spatially and temporally, and (4) a series of illustrative case studies and benchmarks.



\section{Problem Formulation}
\label{sec:problem}


Consider a general dynamical system described by
\begin{equation}
    \label{eq:system}
    \dot x = f(x,u)
\end{equation}
where $x \in X \subset \realnumbers^n$ is the state, $u \in U \subset \realnumbers^c$ is the control, and $f: X \times U \to \realnumbers^n$ is the vector field. We assume $f$ is a continuous (possibly nonlinear) function.  
Let $\mathbb{T} = \realnumbers_{\geq 0}$ denote the continuous time domain.
Then, given an initial state $x_0 \in X$ and a controller $\pu: \mathbb{T} \to U$, the solution of System  \eqref{eq:system} is a state trajectory $\px: \mathbb{T} \to X$.

We are interested in the  properties of System \eqref{eq:system} with respect to a set of predicates defined over state space $X$.
Let $H = \{h_1, h_2, \ldots, h_l\}$ be a given set of functions, where $h_i : \realnumbers^n \rightarrow \realnumbers$ is a polynomial function for every $1 \leq i \leq l$.  
Then, the set of predicate $\Predicate = \{\predicate_1, \cdots, \predicate_l\}$ is defined on $H$ such that, $\forall i \in \{1,\ldots,l\}$, the Boolean value of $\predicate_i: X \to \{\top,\bot\}$ is determined by the sign of function $h_i$ as:
\begin{equation}
    \label{eq: predicate def}
    \predicate_i(x) := 
    \begin{cases}
        \top  & \text{ if } h_i(x) \geq 0\\
        \bot & \text{ if } h_i(x) < 0.
    \end{cases}
\end{equation}
To express the timed temporal properties over $\Predicate$, we use STL. Specifically, we focus on the \textit{STL with non-nested temporal operators} (\stlnn) fragment, which 
is more expressive than the one in \cite{Lindemann2019} and includes all the key temporal operators of STL except for their nesting.

\begin{definition}[STL$_\text{nn}$ Syntax]
    \label{def:STLsyntax}
    The \stlnn syntax is recursively defined by:
    \begin{align*}
        \varphi &:= \phi \mid \neg \varphi \mid \varphi \wedge \varphi \mid \phi \until_I \phi\\
        \phi &:= \predicate \mid \neg \phi \mid \phi \wedge \phi
    \end{align*}
    where $\predicate \in \Predicate$ is a predicate, and $I = \la a,b \ra$ is a time interval with $\la \in \big\{ [, ( \big\}$, $\ra \in \big\{ ], ) \big\}$, $a,b \in \mathbb{T}$, and $a \leq b < \infty$.
\end{definition}

\noindent
From this syntax, one can derive other standard temporal operators, e.g., the time-constrained \textit{eventually}
($\eventually_I$) and \textit{globally} ($\globally_I$) operators are given by
    \begin{align}
        \eventually_I \phi \equiv \top \until_I \phi \quad \text{and } \quad \globally_I \phi \equiv \neg \eventually_I \neg \phi.
    \end{align}

The interpretations of STL formulas are over trajectories of System~\eqref{eq:system} as defined below.
\begin{definition}[\stlnn Semantics]
    \label{def:STLsemantics}
    The semantics of \stlnn formula $\varphi$ is defined over trajectory $\px$ recursively as:
    \begin{itemize}
        \item $(\px,t) \models \top \iff \top$;
        \item $(\px,t) \models \predicate \iff h(\px(t)) > 0$;
        \item $(\px, t) \models \neg \varphi \iff (\px, t)  \not\models \varphi$;
        \item $(\px,t) \models \varphi \wedge \varphi' \iff (\px,t) \models \varphi \wedge (\px,t) \models \varphi'$;
        \item $(\px,t) \models \varphi \until_{\la a,b \ra} \varphi' \iff \exists t' \in \la t + a, t + b \ra$ such that $(\px,t') \models \varphi' \; \wedge \; \forall t'' \in [t, t'], \; (\px,t'') \models \varphi$,
\end{itemize}
    where $(\px,t)$ 
    denotes trajectory $\px$ at time $t$, and $\models$ is the satisfaction relation. Trajectory $\px$ satisfies formula $\varphi$, denoted by $\px \models \varphi$, if $(\px,0) \models \varphi$.
\end{definition}
\noindent

Given a property of interest expressed in \stlnn formula $\varphi$, our goal is to synthesize a controller $\pu$ such that the resulting trajectory $\px$ satisfies $\varphi$.  

\begin{problem}
    \label{prob:controllersynthesis}
    Given System \eqref{eq:system}, a set of polynomial functions $H=\{h_1,\ldots, h_l\}$ and their corresponding set of predicates $\Predicate = \{\predicate_1, \ldots, \predicate_l \}$ per \eqref{eq: predicate def},
    and an \stlnn formula $\varphi$ defined over $\Predicate$, synthesize a controller $\pu: \mathbb{T} \to U$ such that the resulting system trajectory $\px \models \varphi$.
\end{problem}

Note that Problem~\ref{prob:controllersynthesis} poses two major challenges: (i) System~\eqref{eq:system} has general (possibly nonlinear) dynamics, and (ii) the predicates in $\Predicate$ are defined over polynomial functions. Hence, optimization-based as well as linearization-based approaches are likely to suffer due to non-convexity and approximation errors.
Instead, we approach the problem by employing sampling-based techniques to deal with (i) and semi-algebraic set analysis to deal with (ii). To this end, we first develop a method to construct a hybrid automaton from $\varphi$ and System~\eqref{eq:system} to capture the set of all trajectories that satisfy $\varphi$. This reduces the problem to reachability analysis of the hybrid automaton. Then, we perform this analysis efficiently using a guided sampling-based search.

\section{Preliminaries}
\label{sec:preliminaries}
In this section, we formally define concepts that are used in our approach.
We use \textit{LTL with finite traces} (LTLf) in the construction of the automaton from \stlnn formulas.
\begin{definition}[LTLf Syntax]
    Let $P$ be a set of atomic propositions. Then, LTLf syntax is recursively defined by:
    \begin{align*}
        \psi := p \mid \neg \psi \mid \psi \wedge \psi \mid \tempop{X} \varphi \mid \psi \until \psi
    \end{align*}
    where $p \in P$, and $\tempop{X}$ and $\until$ are the \emph{next} and \emph{until} operators.
\end{definition}

The semantics of LTLf are defined over finite words.  See \cite{ltlf} for details.
\noindent
We utilize non-deterministic finite automata (NFA) to guide sampling-based search.

\begin{definition}[NFA]
    \label{def:NFA}
    A non-deterministic finite automaton (NFA) is a tuple $\mathcal{A} = (Q, Q_0, \Sigma,  \delta, F)$, where
    \begin{itemize}
        \item $Q$ is a finite set of states;
        \item $Q_0 \subseteq Q$ is the set of initial states;
        \item $\Sigma$ is a set of input symbols;
        \item $\delta : Q \times \Sigma \rightarrow 2^Q$ is the transition function;
        \item $F \subseteq Q$ is the set of accepting states.
    \end{itemize}
\end{definition}

\noindent
Given a finite word $\pi = \sigma_1 \sigma_2 \ldots \sigma_n$ where $\sigma \in \Sigma$, a run $q = q_0q_1\ldots q_n$, where $q_0 \in Q_0$ and $q_{i+1} \in \delta(q_i, \sigma_{i+1})$ for all $0\leq i<n$, on NFA $\mathcal{A}$ is induced.  We say $\pi$ is accepted by $\mathcal{A}$ if there exists an induced run $q$ such that $q_n \in F$.  The \textit{language} of $\mathcal{A}$ denoted by $\mathcal{L}(\mathcal{A})$ is the set of words that are accepted by $\mathcal{A}$.

A \textit{deterministic finite automaton} (DFA) is an NFA in which $|Q_0|=1$ and, $\forall q \in Q$ and $\forall \sigma \in \Sigma$, $|\delta(q,\sigma)| \leq 1$. 
Every LTLf formula $\psi$ defined over $P$ can be algorithmically translated to a DFA $\mathcal{A}_\psi$ with symbols $\Sigma = 2^P$ such that $\pi \models \psi \iff \pi \in \mathcal{L}(\mathcal{A}_\psi)$.

From \stlnn formulas, we construct an NFA that is constrained with time.
\begin{definition}[Timed NFA]
A timed NFA is a tuple $\TA = (Q, \Sigma_t, Q_0, \text{Inv}, \delta, F)$ where $Q, Q_0, \delta$ and $F$ are as in Def~\ref{def:NFA}, 
    \begin{itemize}
        \item $\Sigma_t = \Sigma \times \mathbb{T}$ is the set of input symbols,
        where each symbol in $\Sigma_t$ is a pair of $\sigma \in \Sigma$ and time $t \in \timedomain$, and
        \item $\text{Inv} : Q \rightarrow \mathcal{I}$ is an invariant function over a time interval, where $\mathcal{I}$ is the set of all intervals over $\mathbb{T}$.
    \end{itemize}
\end{definition}
\noindent
In $\TA$ a transition $q \xrightarrow{(\sigma, t)} q'$ exists iff $q' \in \delta(q, \sigma)$ and $t \in \text{Inv}(q')$.





\begin{definition}[Timed path]
    Let $L: X \rightarrow \Sigma$ be a labeling function that maps a state $x  \in X$ to an input symbol $\sigma \in \Sigma$.
    A finite trajectory $\px$ in a state space $X$ defines a timed path $(\mathbf{p}, \mathbf{t}) = (x_0, t_0)(x_1, t_1)\cdots(x_l,t_l)$ such that for all $i \in \{0, \cdots, l\}$:
    \begin{itemize}
        \item $\mathbf{p} = x_0x_1\cdots x_l$, with $x_i = \px(t_i)$,
        \item $t_0t_1\cdots t_l$ is a sequence of time stamps, with $t_i \in \timedomain$ and $t_{i+1} \geq t_{i}$,
        \item $L(x_i)$ remains constant during the time interval $(t_i, t_{i+1})$, i.e., $L(\px(t)) = L(x_i)$ $\forall t \in (t_i,t_{i+1})$.
    \end{itemize}
\end{definition}

\begin{definition}[Timed word]
    A finite trajectory $\px$ with a timed path $(\mathbf{p}, \mathbf{t})$ generates a timed word $w(\px) = w((\mathbf{p}, \mathbf{t})) = (L(x_0), t_0)(L(x_1), t_1)\cdots(L(x_l),t_l)$.
\end{definition}

\noindent
The notions of acceptance and language of $\TA$ over timed words are extended from NFAs in a straightforward manner.  In addition, we say that a timed path is accepted by $\TA$ if its timed word is accepted by $\TA$.

\section{Temporal-aware Automata Guided Search}
\label{sec:method}

Here, we present our control synthesis framework that solves Problem~\ref{prob:controllersynthesis}. Our approach can be formalized by decomposing Problem~\ref{prob:controllersynthesis} into two steps: timed automaton construction and sampling-based search.
We first observe that \stlnn formulas reason over a single (global) time. Taking advantage of this, we propose a novel algorithm to efficiently construct a timed NFA $\TA$ such that every trajectory that satisfies the \stlnn formula has a timed word that is accepted by $\TA$.
Next, given $\TA$, our goal is to synthesize a controller $\pu$ that drives System~\eqref{eq:system} through a trajectory that satisfies $\varphi$. Note this is a reachability problem on the hybrid system that is induced from $\TA$ and System~\eqref{eq:system}. To efficiently solve this problem, we employ a sampling-based approach with high-level guides using an automaton. An overview of the algorithm is shown in Alg.~\ref{alg:framework}.


\begin{algorithm}
    \caption{Automaton-Guided Synthesis for STL}
    \label{alg:framework}
    \SetKwInOut{Input}{Input}\SetKwInOut{Output}{Output}
    \Input{Formula $\varphi$, $X$, System~\eqref{eq:system}, $x_0$, $t_{max}$, $t_{e}$.}
    \Output{Trajectory and controllers if successful, otherwise NULL.}
    $\mathcal{M} \leftarrow$ ComputeAbstraction$(X, \varphi.H)$ \\
    $\TA_\varphi \leftarrow$ ConstructAutomaton$(\varphi)$\\
    $\mathcal{P} \leftarrow \mathcal{M} \otimes \TA_\varphi$\\
    $\textsc{tr} \leftarrow$ InitializeTree($x_0$)\\
    \While{Time Elapsed $< t_{max}$}
    {
        $\mathbf{L} \leftarrow \{(d_i, q_i)\}_{i \geq 0} \leftarrow$ ComputeLead($\mathcal{P}, t_{max}$)\\
        $C \leftarrow$ ComputeAvailableProductStates($\textsc{tr}, \mathbf{L}$)\\
        $(\mathbf{v}, \mathbf{u}, \mathbf{t}) \leftarrow$ Explore($\mathbf{L}, \textsc{tr}, \mathcal{P}, X, t_{e}$)\\
        \If{$(\mathbf{v}, \mathbf{u}, \mathbf{t}) \neq \emptyset$}{
            \Return $(\mathbf{v}, \mathbf{u}, \mathbf{t})$
        }
    }
    \Return No Solution
\end{algorithm}



\subsection{State Predicate Abstraction}
\label{sec:abstraction}
To construct $\TA$, we first establish a relationship between predicates and state labels. To this end, we abstract the system into a discrete model based on the predicates. Note that predicates in $\Predicate$ define a semi-algebraic set on $\realnumbers^n$ since every $h \in H$ is a polynomial function. Hence, bounded state space $X \subset \realnumbers^n$ can be partitioned into a set of disjoint regions $D$ (i.e., the pairwise intersection of every two $d, d' \in D$ is empty, and the union of all $d \in D$ is the state space $X$), such that the points in each region are predicate-truth invariant, i.e., have the same truth-value assignments of predicates in $\Predicate$.  Formally, for a partition region $d \subseteq X$ and every two points $x,x' \in d$, $\predicate(x) = \predicate(x')$ for all $\predicate \in \Predicate$. Through predicate abstraction, each state $x \in X$ is associated with a subset of $\Predicate$, a symbol $\sigma \in \Sigma = 2^{\Predicate}$. 
We define $L : D \rightarrow 2^{\Predicate}$ to be a labeling function that assigns to $d \in D$ the predicates that are true for all $x \in d$.

Predicate abstraction can be achieved in multiple ways. For polynomial predicate functions, methods such as cylindrical algebraic decomposition \cite{Lindemann2007, schwartz1983} can be used to find decompositions for predicate invariance. Computing this decomposition has a worst-case complexity of doubly exponential in the dimension of the space. For linear predicate functions, there exist efficient algorithms, e.g., geometry-based conforming Delaunay triangulation, to construct partitions that respect the predicate boundaries.

After partitioning the space into regions with predicate-truth invariance, we use geometric adjacency between the regions to construct an abstraction graph. This abstraction graph is a tuple $\mathcal{M} = (D, E, \Predicate, L)$, where the set of discrete states is $D$, and an edge $(d, d') \in E$ exists if there is adjacency between $d \in D$ and $d' \in D$.
Note that $\mathcal{M}$ captures all possible transitions between adjacent regions even though they may not be realizable by System~\eqref{eq:system}. Such transitions are implicitly ignored in the second planning step of our framework, described in Sec.~\ref{sec:planning}. This is an advantage of our approach since it drastically reduces the computational burden of constructing a highly accurate abstraction.
\subsection{Automaton Construction}

Here, we introduce a novel method to construct a timed NFA from \stlnn formulas. This construction takes advantage of the non-nested aspect of this fragment for syntactic separation, and uses techniques from automata theory to construct an automaton that is partitioned in the time dimension. 

\begin{algorithm}
    \caption{Construct Automaton}
    \label{alg:automatonconstruction}
    \SetKwInOut{Input}{Input}\SetKwInOut{Output}{Output}
    \Input{\stlnn formula $\varphi$, time partition set $\mathbf{T}$.}
    \Output{Timed NFA $\TA_{\varphi}$.}
    $\{\Phi_i\}_{\geq 0} \leftarrow$ TimePartition($\varphi, \mathbf{T}$), $\TA \leftarrow \emptyset$\\
    \For{$\Phi \in \{\Phi_i\}_{i \geq 0}$}
    {
        \For{$\psi_{(t_{j-1}, t_j]} \in \Phi$}{
            $\psi^j \leftarrow$ IgnoreTimeWindows($\psi_{(t_{j-1}, t_j]}$)\\
            $\mathcal{A}_j \leftarrow $ LTLf2DFA($\psi^j$)\\
            $\TA_j \leftarrow$ ConvertToTimedDFA($\mathcal{A}_{j}, \psi_{(t_{j-1}, t_j]}$)\\
            $\TA_{\varphi} \leftarrow$ ConnectDFA($\TA_{\varphi}, \TA_{j}$)
        }
    }
    \Return $\TA_\varphi$
\end{algorithm}

The algorithm is shown in Alg.~\ref{alg:automatonconstruction}. It takes an \stlnn formula $\varphi$ and time partition set $\mathbf{T}$, and returns a timed NFA.  We explain the steps of the algorithm below and use the following running example to clarify the key concepts.

\begin{example}[Reach-avoid]
    \label{ex:reachavoidstl}
    Consider a system with $X \subset \realnumbers^2$ and specification ``reach a goal region defined by $3 < x_1 \leq 4$ and $2 < x_2 \leq 3$ in the first 18 seconds while avoiding unsafe region $3 < x_1 \leq 4$ and $2 < x_2 \leq 3$ in the first 6 seconds":
    \begin{align*}
        \varphi_\text{RA} =& \eventually_{[0,18]}(3 < x_1 \leq 4 \wedge 2 < x_2 \leq 3) \;\; \wedge \\ &  \globally_{[0,6]} \neg (1 < x_1 \leq 2 \wedge 2 < x_2 \leq 3) \\
        =& \eventually_{[0,18]}(\predicate_1(\px)) \wedge \globally_{[0,6]} \neg (\predicate_2(\px)).
    \end{align*}
\end{example}

\subsubsection{Time partitioning}


Given $\varphi$, our goal is to obtain an equivalent formula in which subformulas depend only on a disjoint segment of time. We can use the following equivalence rule to perform a syntactic separation.

\begin{proposition}[STL Separation \cite{boundedmodelchecking}]
    \label{prop:stlseparation}
    A \textit{time separation} of $\varphi \until_{[a,b]}\varphi'$ at time $\tau \in [a, b]$ is given by:
    \begin{align*}
    \varphi \until_{[a,b]}\varphi' \equiv &\; \varphi \until_{[a,\tau)} \varphi'  \nonumber \\
    &\vee \Big( \globally_{[a,\tau)} \varphi \wedge \big( \eventually_{[\tau,\tau]}(\varphi \wedge \varphi') \vee \varphi \until_{(\tau, b]}\varphi' \big) \Big).
\end{align*}

\end{proposition}
\noindent
Note that Proposition~\ref{prop:stlseparation} focuses on $\varphi \until_{[a,b]}\varphi'$, where $I = [a,b]$ is a closed time interval. Other cases of $I = \la a, b \ra$ are not shown but can be derived in a similar manner. 

We use this equivalence rule to syntactically separate the \stlnn formula $\varphi$ into $k$ time partitions, at time points in $\mathbf{T} = \{0,\tau_1, \cdots, \tau_{k-1}\}$ to obtain $k$ subformulas $\varphi_1, \cdots, \varphi_k$. We then express the partitioned formula in Disjunctive Normal Form (DNF) and obtain $\varphi = \Phi_1 \vee \Phi_2 \vee \cdots \vee \Phi_n$, where $\Phi_i, i \in \{1, \cdots, \Phi_n\}$ are clauses of DNF $\varphi$.

The role of syntactic separation is to reformulate $\varphi$ into time discretized subformulas. By syntactically separating a formula into multiple sub-formulas connected by time, we are able to capture an increased time granularity in the timed NFA, as discussed in Sec.~\ref{sec:experiments}.

We define the minimal time partition as the minimum number of partitions necessary to ensure that the time intervals are equivalent for all temporal operators in each subformula of $\Phi_i$. Note that this can always be achieved because of the non-nested property of \stlnn.

\begin{definition}[Minimal Time Partition]
    Let $T$ be the set of boundary points of the time intervals of the temporal operators in $\varphi$. The minimal time partition is $\mathbf{T} = T \setminus \max(T)$.
\end{definition}

\begin{example}
\label{ex:subformulas}
Consider $\varphi_\text{RA}$ in Example~\ref{ex:reachavoidstl}. Using syntactic separation, we can separate each subformula at time $\mathbf{T} = \{0,6\}$ to form the minimal time partition, and with straightforward manipulation, we obtain the following formula in DNF form:
\vspace{-1mm}
    $\varphi_\text{RA} = \Phi_1 \vee \Phi_2 \vee \Phi_3,$
    where  
    \vspace{-0.5mm}
    \begin{align*}
        \Phi_1 =    & \eventually_{[0,6]}(\predicate_1(\px)) \wedge \globally_{[0,6]}\neg (\predicate_2(\px)), \\
        \Phi_2 =    &  \globally_{[0,6)} \neg (\predicate_2(\px)) \wedge \globally_{[6,6]}\neg (\predicate_2(\px)) \wedge \eventually_{[6,6]}(\predicate_1(\px)), 
        \\
        \Phi_3 =    & \globally_{[0,6]} \neg (\predicate_2(\px)) \wedge \eventually_{(6,18]}(\predicate_1(\px)).
    \end{align*}
\end{example}

\subsubsection{Parse Tree}

From a time partitioned formula $\varphi$, we can construct a parse tree $\mathcal{T}$ with $n$ branches, each corresponding to a $\Phi_i$, $i\in \{1,\ldots,n\}$. Without loss of generality, let subformula $\Phi_i$ 
be in the following form:
\vspace{-.5mm}
\begin{align}
    \label{eq:hybridautomata}
    \Phi_i = \psi_{\la 0,t_1 \ra}\wedge\psi_{\la t_1,t_2 \ra}\wedge\cdots\wedge\psi_{\la t_{m-1},t_{m} \ra}.
\end{align}

\vspace{-0.5mm}
\noindent
$\Phi_i$ can always be written in this form by aggregating subformulas with the same time interval. Then, in the branch of $\mathcal{T}$ corresponding to $\Phi_i$, there are 
$m$ nodes, each representing a time-partitioned subformula $\psi_{\la t_{j-1},t_j \ra}$ for $j \in \{1,\ldots,m\}$.
\begin{example}[Parse Tree]
    \label{ex:parsetree}
    Consider the partitioned formula in Example~\ref{ex:subformulas}. 
    Each $\Phi_i$ can be represented as a branch
    $\mathcal{B}_i = (\text{V}_i,E_i)$ consisting of a set of nodes $\text{V}_i$ and edges $E_i$:
    \vspace{-1mm}
    \begin{align*}
        \mathcal{B}_1 = & (\{\text{v}_{\psi^{1}_{[0,6]}}\}, \emptyset) \text{ from } \Phi_1,\\
        \mathcal{B}_2 = & (\{\{\text{v}_{\psi^{2}_{[0,6)}}, \{\text{v}_{\psi^2_{[6,6]}}\},\{(\{\text{v}_{\psi^{2}_{[0,6)}}, \{\text{v}_{\psi^2_{[6,6]}})\} \text{ from } \Phi_2,\\
        \mathcal{B}_3 = & (\{\{\text{v}_{\psi^{3}_{[0,6]}}, \{\text{v}_{\psi^3_{(6,18]}}\}, \{(\{\text{v}_{\psi^{3}_{[0,6]}}, \{\text{v}_{\psi^3_{(6,18]}})\}) \text{ from } \Phi_3.
    \end{align*}
    
    \vspace{-1mm}
    \noindent
    Then, the parse tree is $\mathcal{T} = (\text{V}_1 \cup \text{V}_2 \cup \text{V}_3, E_1 \cup E_2 \cup E_3)$ as shown in Fig~\ref{fig:parsetree}.
\end{example}

\begin{figure}[t]
    \centering
    \resizebox{0.4\linewidth}{!}{
    \begin{tikzpicture}[->,>=stealth]
    \node[circle, fill,scale=0.3] (s0) at(0,0) {};
      \node[state, above right= of s0] (q1) {$\psi^{1}_{[0,6]}$};
      \node[state, right = 0.5cm of s0] (q3) {$\psi^2_{[0,6)}$};
      \node[state, right = 0.5cm of q3] (q4) {$\psi^2_{[6,6]}$};
      \node[state, below right = of s0] (q5) {$\psi^3_{[0,6]}$};
      \node[state, right = 1cm of q5] (q6) {$\psi^3_{(6,18]}$};
        \draw (q3) edge[above] node{} (q4)
        (s0) edge[above] node{} (q1)
        (s0) edge[above] node{} (q3)
        (s0) edge[above] node{} (q5)
        (q5) edge[above] node{} (q6);
    \end{tikzpicture}
    }
    \caption{\small Parse Tree in Example \ref{ex:parsetree}.
    }
    \label{fig:parsetree}
    \vspace{-5mm}
\end{figure}
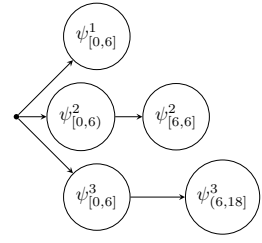

\subsubsection{Timed NFA Construction}
Now, consider a parse tree node corresponding to formula $\psi_{\la t_{j-1},t_{j} \ra}$. By ignoring the time interval 
$\la t_{j-1},t_j \ra$, we obtain the untimed version of this formula, $\psi^{j}$. By using the labeling scheme described above to define atomic propositions on predicates, $\psi^{j}$ can be treated as an LTLf formula since (i) the evaluation of $\psi_{\la t_{j-1}, t_j \ra}$ is over the same time interval for all temporal operators in $\psi^{j}$ due to syntactic separation, and (ii) each time interval is finite. Therefore, if a timed path $(\mathbf{p}, \mathbf{t}) \models \psi_{\la t_{j-1},t_j \ra}$, then its untimed word $w(\mathbf{p}) \models \psi^{j}$.

Recall that an LTLf formula can be translated into a DFA with an equivalent language. Thus, we can translate $\psi^{j}$ to a DFA $\mathcal{A}_{j}$ (for every node of the parse tree). To respect the time interval of the original formula $\psi_{\la t_{j-1},t_j \ra}$, we add the time interval $\la t_{j-1},t_j \ra$ as an invariant to every state in $\mathcal{A}_j$, obtaining timed DFA $\TA_j$ with the set of symbols $2^\Predicate \times \mathbb{T}$.  


To complete the construction of the automaton $\TA_\varphi$ for the original $\varphi$, we connect the states of the timed DFAs successively according to the edges in $\mathcal{T}$.  Let $\mathcal{A}_j$ and $\mathcal{A}_{j+1}$ be two DFAs corresponding to a parent and its child node $\mathcal{T}$, respectively.  Further, let $q^F_j$ be an accepting state of $\mathcal{A}_j$ and $q^0_{j+1}$ be the initial state of $\mathcal{A}_{j+1}$. We add a transition from $q^F_j$ to every successor of $q^0_{j+1}$ with the corresponding symbol in $2^\Gamma$. Finally, we define the accepting states of $\TA_\varphi$ using the accepting states of the DFAs corresponding to the leaf nodes of the parse tree $\mathcal{T}$ in the following way. If the accepting state of this DFA is an absorbing state, then the same corresponding state is accepting. If this state is not absorbing, we add a transition to a new absorbing accepting state $q^F_\varphi$ with the time interval ${\la t_f, \infty)}$ 
as an invariant, where $t_f$ is the right time bound of the previous accepting state. The resulting structure is $\TA_\varphi$, which is a single timed NFA due to the added transitions between the DFAs.

\vspace{-1mm}
\begin{example}
 \label{ex:hybridautomaton}
Consider subformula 
    $\Phi_3 = \globally_{[0,6]} \neg (\predicate_2(\px)) \wedge \eventually_{(6,18]}\predicate_1(\px))$ from Example~\ref{ex:subformulas}. Using \eqref{eq:hybridautomata}, we obtain
    \begin{align*}
        \psi^3_{[0, 6]} = \globally_{[0,6]} \neg (\predicate_2(\px)), \quad
        \psi^3_{(6,18]} = \eventually_{(6,18]}(\predicate_1(\px)).
    \end{align*}
    We translate each $\psi$ to an LTLf formula by evaluating the untimed version of $\psi$: \quad
    \begin{align*}
        \psi^{0,6} = \globally \neg (\predicate_2(\px)), \quad 
        \psi^{6,18} = \eventually (\predicate_1(\px)),
    \end{align*}
    and translate each LTLf formula into a DFA. The DFA for each LTLf formula is shown in Fig.~\ref{fig:dfa_0} and Fig. \ref{fig:dfa_6}.
    The final timed automaton $\TA_{\Phi_3}$ is show in Fig. \ref{fig:final_TA}.
\end{example}

\begin{figure}[t]
  \begin{center}
    \begin{subfigure}{0.49\linewidth}
        \centering
        \scalebox{0.65}{
        \begin{tikzpicture}[->,>=stealth]
      \node[state, initial, accepting] (q1) {$q_1$};
        \draw (q1) edge[loop above] node{$\neg\predicate_2$} (q1);
    \end{tikzpicture}
    }
    \caption{\small DFA for $\psi^{0,6}$ 
    \label{fig:dfa_0}
    \vspace{2mm}
    }
    \end{subfigure}
    \begin{subfigure}{0.49\linewidth}
        \centering
        \scalebox{0.65}{
        \begin{tikzpicture}[->,>=stealth]
            \node[state, initial] (q1) {$q_1$};
            \node[state, accepting, right= 1cm of q1] (q2) {$q_2$};
            \draw (q1) edge[loop above] node{$\neg \predicate_1$} (q1)
            (q1) edge[above] node{$\predicate_1$} (q2)
            (q2) edge[loop above] node{$\top$} (q2);
        \end{tikzpicture}
        }
    \caption{\small DFA for $\psi^{6,18}$
    \label{fig:dfa_6}
    }
    \vspace{2mm}
    \end{subfigure}
    \newline
    \begin{subfigure}{0.9\linewidth}
            \centering
            \scalebox{0.65}{
            \begin{tikzpicture}[->,>=stealth]
            \node[state, initial] (q1) {\shortstack{$q_1$\\ $t \in [0, 6]$}};
            \node[state, right = 2.5cm of q1] (q2) {\shortstack{$q_2$\\ $t \in (6, 18]$}};
            \node[state, accepting, below = of q2] (q3) {\shortstack{$q_3$\\ $t \in (6, 18]$}};
            \draw (q1) edge[loop below] node{$\neg\predicate_2$} (q1)
            (q1) edge[above] node{$\neg \predicate_1$} (q2)
            (q1) edge[bend right, below] node{$\predicate_1$} (q3)
            (q2) edge[loop right] node{$\neg \predicate_1$} (q2)
            (q2) edge[right] node{$\predicate_1$} (q3)
            (q3) edge[loop right] node{$\top$} (q3);
            \end{tikzpicture}
            }
            \caption{\small Final timed NFA. State $q_1$ corresponds to the DFA in (a) and $q_2$ and $q_3$ correspond to the DFA in (b).}
            \label{fig:final_TA}
        \end{subfigure}
    \caption{\small Timed NFA construction of $\Phi_3$ for Example~\ref{ex:hybridautomaton}.}
    \label{fig:dfa}
  \end{center}
  \vspace{-4mm}
\end{figure}
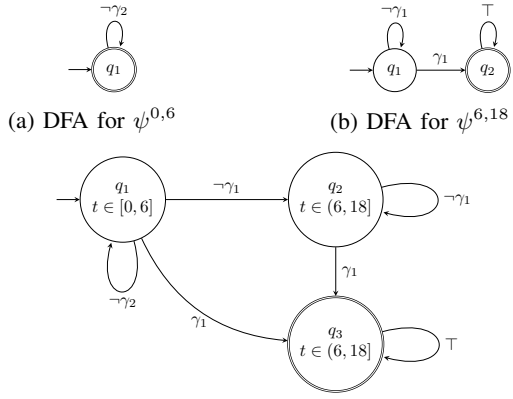

The following theorem
states that the constructed timed automaton accepts precisely the timed words of the trajectory that satisfy $\varphi$.
\vspace{-1mm}
\begin{theorem} 
\label{theorem:timedsatisfiesstl}
A trajectory (timed path) is accepted by the timed NFA $\TA_\varphi$ if and only if it satisfies \stlnn formula $\varphi$.
\end{theorem}
\noindent 
The proof of Theorem \ref{theorem:timedsatisfiesstl} is in the extended version of this work \cite{CDC22extendedversion}.

\subsection{Product Automaton}

The final structure used for high-level discrete search is a product automaton $\mathcal{P}$, which captures all possible ways that paths can satisfy $\varphi$. The product automaton is computed as $\mathcal{P} = \TA_\varphi \otimes \mathcal{M}$ (Line 3 of Alg.~\ref{alg:framework}). $\mathcal{P}$ is a tuple $(Z, Z_0, \delta_p, \text{Inv}_p, F_p)$, where
$Z = Q \times D$, $F_p = F \times D$, 
\begin{itemize}
    \item $\delta_p : Q \times D \rightarrow Q \times D$, such that $(q', d') \in \delta_p(q, d)$ if $(d, d') \in E$ and $q' \in \delta(q, L(d))$, and
    \item $\text{Inv}: Q \times D \rightarrow \mathcal{I}$, such that  $\text{Inv}((q,d)) = \text{Inv}(q)$.
\end{itemize}

We refer to state $z = (q,d) \in Z$ of $\mathcal{P}$ as a high-level state. 
State $z$ inherits the invariant of its automaton state $q$. State $z$ is accepting if $q \in F$ is an accepting state.
Every accepting run of $\mathcal{P}$ is a sequence of high-level states, $(q_0, d_0)(q_1, d_1)\cdots(q_p, d_p)$. A continuous trajectory that follows the sequence of regions $d_0d_1\cdots d_p$ and respects the invariant constraints of $q_0q_1\cdots q_p$ satisfies $\varphi$. 



\subsection{Time Partitioned Automaton-guided Search}
\label{sec:planning}
To generate a satisfying trajectory, we extend the co-safe LTL planning framework in \cite{Bhatia2010, Maly2013} to include time window constraints.
The framework consists of two main layers: an automaton-based high-level search layer and a low-level sampling-based planner. 

\subsubsection{High-level layer}
\label{sec:highlevel}
In each iteration of high-level planning, the high-level layer searches the product automaton $\mathcal{P}$ constructed above. By searching $\mathcal{P}$, an accepting run (called lead) $\mathbf{L}$ is computed at each iteration. Each lead is a
candidate sequence of contiguous high-level states through
which the algorithm attempts to guide a sampling-based motion tree of continuous trajectories in the low-level layer. A graph search algorithm, such as Dijkstra's Algorithm, is used to find the shortest path to an accepting state from a start state according to the following edge weight computation. For each high-level state $(q,d)$, we propose to assign a weight defined by
\begin{align}
    w(d,q) = \frac{(cov(q,d) +1)\cdot vol(d) \cdot duration(q)}{DistFromAcc(q) \cdot (numsel(q,d) +1)^2},
\end{align}
where $cov(q,d)$ is the number of motion tree vertices in the state $(q,d)$, $vol(d)$ is the estimated volume of the state space covered by region $d$, $duration(q)$ is size of the time interval associated with automaton state $q$, $DistFromAcc(q)$ is the shortest unweighted path from the automaton state $q$ to an accepting state in $\TA_\varphi$, and $numsel(q,d)$ is the number of times ($q,d$) has been selected for exploration. The edge weight between two states is then defined as:
\begin{align}
    w((q_1,d_1),(q_2,d_2)) = \frac{1}{w(q_1,d_1)\cdot w(q_2,d_2)}.
\end{align}
These weights are updated in every iteration. Such a weighting scheme has been shown to promote expansion in unexplored areas (i.e., where $cov$ and $numsel$ are both small) and to suppress expansion in areas where attempts at exploration have repeatedly failed (i.e., where $numsel \gg cov$) \cite{Maly2013}.
\subsubsection{Low-level layer - Incorporating Time Constraints}

\begin{algorithm}
    \caption{Explore}
    \label{alg:explore}
    \SetKwInOut{Input}{Input}\SetKwInOut{Output}{Output}
    \Input{$\mathbf{L}$, $C$, $\textsc{tr}$, $\mathcal{P}$,$X$, $t_{e}$,}
    \Output{Trajectory of tree nodes $(\mathbf{v},\mathbf{t})$ and controller $\pu$, if successful; return NULL otherwise.}
    \While{Time Elapsed $< t_e$}
    {
        $(d,q) \leftarrow$ C.sample()\\
        $v_s \leftarrow $ Sample(d,q)\\
        $v \leftarrow $ SelectAndExtend($\mathcal{P}, \textsc{tr}, v_s, X$)\\
        \If{isValid($v$)}{
        Add vertex $v$ to \textsc{tr}\\ 
        \If{isAccepting($v.q$)}{
            $(\mathbf{v},\mathbf{u}, \mathbf{t})  \leftarrow$ Extract trajectory ending with $v$ and control inputs from \textsc{tr}\\
            \Return $(\mathbf{v}, \pu, \mathbf{t})$
        }
        \If{$(v.d, v.q) \notin C \wedge (v.d, v.q) \in L$}
        {
            $C \leftarrow C \cup \{(v.d, v.q)\}$
        }
        }
    }
    \Return $\emptyset$
\end{algorithm}
Given a lead $\mathbf{L}$, we compute the set of high-level states $C$ that exist in $\mathbf{L}$ and contain at least one vertex of a motion tree (line 7 of Alg.~\ref{alg:framework}). The low-level search promotes tree expansion in the high-level states of $C$ (Alg.~\ref{alg:explore}). In each iteration of low-level planning, a high-level state $z$ is sampled from $C$ with probability proportional to the weights of each high-level state in $C$. A tree vertex $v_s$ is then sampled within $z$, and one iteration of tree-search motion planning is performed to obtain a new tree vertex $v_n$. A tree vertex is a tuple $v = (x, q, d, t)$. An iteration of tree-search motion planning involves selecting a tree vertex $v_{sel}$, sampling a valid control $u \in U$ and time duration, and propagating the dynamics \eqref{eq:system} of the continuous state component $x_{sel}$ of $v_{sel}$ using numerical propagation techniques (e.g., Runge-Kutta methods) to obtain a new continuous state component $x_n$ and time $t_n$. Any tree-based motion planner that supports kinodynamic constraints can be used at this step, such as RRT \cite{kinorrt}. The new discrete state components $q_n, d_n$ of $v_n$ are obtained by propagating the product automaton $\mathcal{P}$.

Vertex $v$ is valid if both its continuous and discrete state components are valid. The continuous component $v.x$ is valid if it obeys the state constraints of System~\eqref{eq:system}. Our main addition to the existing multi-layered planner is the handling of time constraints in this low-level layer. The discrete components $v.q$ and $v.d$ are valid if $(v.q,v.d)$ exists in the product automaton and the vertex time $v.t$ is in the invariant of $v.q$, i.e., $v.t \in \text{Inv}(v.q)$. Without considering the time component, a newly extended vertex of the motion tree may have multiple possible discrete components (high-level state) due to the non-determinism of $\TA$. The time component determines the valid high-level state.

Suppose that a new vertex $v_n$ corresponds to a newly reached high-level state that is in the current discrete lead $\mathbf{L}$, this high-level state is added to the set of available states for exploration $C$. Finally, if the automaton state $v_n.q$ is an accepting state, the trajectory ending with $v_n$ and its corresponding controller $\pu$ is returned as a solution.




\section{Correctness and Completeness}

In this section, we show that our framework is sound and probabilistically complete.

\begin{lemma}[Correctness]
\label{lemma:correctness}
Given a lead $\mathbf{L}$ computed by $\mathcal{P}$, $\exists \mathbf{t} = t_0t_1\cdots t_l$ associated with $\mathbf{L}$ that produces a timed word $(\mathbf{w}, \mathbf{t})$ such that any trajectory that realizes this timed word satisfies STL$_{nn}$ formula $\varphi$.
\end{lemma}
\begin{proof}[Proof Sketch]
    Observe that $\TA_\varphi$ is an NFA with time constraints.
    The proof is straightforward from the definition of the language of the NFA corresponding to $\TA_\varphi$ and Theorem~\ref{theorem:timedsatisfiesstl}. Consider a lead $\mathbf{L}$ that produces a word $\mathbf{w}$, through the sequence of regions $\mathbf{d}$, accepted by $\mathcal{P}$. Since $\mathcal{P}$ is a Cartesian product of $\TA_\varphi$ and $\mathcal{M}$, the set of words that are accepted by $\mathcal{P}$ is a subset of the language of the NFA. Therefore, this word can be converted to a timed word $(\mathbf{w}, \mathbf{t})$ that is accepted by $\TA_\varphi$. By Theorem~\ref{theorem:timedsatisfiesstl}, all trajectories that realizes this timed word necessarily satisfies $\varphi$.
    \vspace{-1mm}
\end{proof}
Lemma~\ref{lemma:correctness} provides the correctness guarantee of the proposed framework.  That is, the returned trajectory by the algorithm is guaranteed to satisfy $\varphi$. The algorithm is also probabilistically complete, as given by the theorem below.
\begin{theorem}[Probabilistic Completeness]
    Alg.~\ref{alg:framework} is probabilistically complete, i.e., Alg.~\ref{alg:framework} finds a satisfying solution if one exists almost surely (with probability $1$) as time approaches infinity.
\end{theorem}
\begin{proof}[Proof Sketch]
    The proof of probabilistic completeness follows from the probabilistic completeness of the tree search based algorithm. In each iteration of high-level planning discrete search of the product automaton computes a lead, which is a valid candidate series of high-level regions by Theorem~\ref{lemma:correctness}. For the low-level layer, the probability of finding a continuous trajectory solution that obeys the invariant constraints according to a given lead, if one exists, approaches $1$ as $t_e \rightarrow \infty$. Since the edge weights of the product automaton are updated every iteration, the probability of searching all possible solution paths in $\mathcal{P}$ approaches $1$ as well. By iterating between high-level and low-level planning, the algorithm produces continuous trajectories covering the entire search space with probability $1$ as $t_{\max} \rightarrow \infty$.
\end{proof}
\section{Evaluation}
\label{sec:experiments}
The objective of our evaluation is to analyze the efficacy of our automaton-guided approach for control synthesis under \stlnn specifications. We first illustrate the advantage of our framework in comparison benchmarks against related work \cite{stl-rrtstar} on linear and linearizable systems with linear predicates. Then, we show the impact of time partitioning on algorithm performance. Finally, we demonstrate the generality of our approach on a nonlinear system with polynomial predicates.

\subsection{Comparison Benchmarks}

We first consider the two systems in \cite{stl-rrtstar}, a 2D double integrator $\dot{x} = y$, $\dot{y} = u$
with a specification 
\begin{align*}
    \varphi_1 = &\eventually_{[2,10]}(3.5 < x \leq 4 \wedge -0.2 < y \leq 0.2) \;\; \wedge \\& \globally_{[0,2]}(-0.5 < y \leq 0.5) \;\; \wedge \\& \globally_{[0,10]}((2 < x \leq 3) \rightarrow (y > 0.5 \vee x \leq -0.5)),
\end{align*}
and a second order unicycle system with dynamics
\begin{align*}
    \dot{x} = v \cos{(\phi)}, \quad
    \dot{y} = v \sin{(\phi)}, \quad \dot{\phi} = \omega,  \quad \dot{v} = a\text{,}
\end{align*}
and specification 
    $\varphi_2 = \eventually_{[0,18]}(3 < x \leq 4 \wedge 2 < y \leq 3) \; \wedge \;
               \globally_{[0,6]}\neg(1 < x \leq 2 \wedge 2 < y \leq 3).$

We compared our approach to \cite{stl-rrtstar}, which we call STL-RRT*, another sampling based algorithm for STL that does not utilize automaton construction. We implemented both algorithms in the Open Motion Planning Library (OMPL) \cite{sucan2012the-open-motion-planning-library}. Table~\ref{tab:benchmarking} shows the time taken to compute a satisfying solution for both systems and specifications over $100$ trials, and the number of solutions found within the time limit of 30 seconds. It can be seen that our approach computes solutions considerably faster than STL-RRT*, with a success rate of $100\%$, whereas STL-RRT* can solve only $13\%$ of the instances for the unicycle model.  Additionally, we note that we use the full nonlinear dynamics for the unicycle model in our algorithm, while we need to use state-feedback linearization to ensure the availability of a steering function for both state and time (punctuality) for STL-RRT*. 
\begin{table}[h]
    \centering
    \caption{\small Benchmarking results.}
    \label{tab:benchmarking}
    \scalebox{0.96}{
    \begin{tabular}{l|c|c|c|c}
     & \multicolumn{2}{c|}{Double Integrator} & \multicolumn{2}{c}{Unicycle}\\
    Algorithm &  Succ. (\%) & Time (s) & Succ. (\%) & Time (s)\\\hline
    STL-RRT* & $\mathbf{100}$ & $1.83 \pm 0.12$ & $13$ & $7.51 \pm 0.71$ \\
    $\TA$-based & $\mathbf{100}$ & $\mathbf{1.10 \pm 0.02}$ & $\mathbf{100}$ &  $\mathbf{1.57 \pm 0.11}$
    \end{tabular}
    }
\end{table}



\subsection{Impact of Time partitioning}

We study the benefit of varying the number of time partitions on the double integrator with the specification:
\vspace{-1mm}
\begin{align*}
    \globally_{[0,t_1)} (x \leq 3.5 \wedge y \leq 3.5) \wedge \eventually_{[t_1,t_2]}(x > 3.5 \wedge y > 3.5),
\end{align*}

\vspace{-1mm}
\noindent
which forbids $x > 3.5$ and $y > 3.5$ until $t_1$, after which the system must eventually reach $x > 3.5, y > 3.5$ within the interval $[t_1, t_2]$. The minimal partition for this formula is $\{0,t_1\}$. We further partition the formula between $[0,t_1]$ to varying number of partitions $k$. The success rate and mean time to solution over 100 trials is reported in Table~\ref{tab:partitioning}, given a time limit of 30 seconds. 
\begin{table}[t]
    \centering
    \caption{\small Impact of varying the number of time partitions on solution time and success rate.}
    \vspace{-1mm}
    \begin{tabular}{c|c|c}
    $k$ time partitions & Success (\%) & Time (s) \\\hline
    $1$ & $48$ & $3.77 \pm 0.64$\\
    $2$ & $81$ & $4.99 \pm 0.56$\\
    $\mathbf{5}$ & $\mathbf{100}$ & $\mathbf{4.42 \pm 0.51}$\\
    $15$ & $100$ & $8.35 \pm 0.32$\\
    $30$ & $100$ & $10.50 \pm 0.59$\\
    $50$ & $100$ & $8.89 \pm 0.50$\\
    $200$ & $51$ & $22.7 \pm 0.57$
    \end{tabular}
    \label{tab:partitioning}
    \vspace{-3.5mm}
\end{table}

The results show that varying time partitioning affects the decomposition granularity of the constructed automaton in the time dimension, which directly impacts the computational efficiency of the algorithm. Increasing the number of time partitions can improve the efficiency of solutions. This is due to the high-level planning layer being able to provide guides of finer granularity in the time dimension for the low-level layer. However, decomposition granularity that is too fine can lead to increased computation times, since the lead becomes harder for the low-level layer to follow.

\subsection{Nonlinear system and predicates}

To demonstrate the generality of our framework compared to previous approaches, we used our framework to synthesize a controller for the car-like nonlinear system in \cite{Webb2012KinodynamicRO}, with the following specification involving nonlinear predicates:
\vspace{-1mm}
\begin{align*}
    \varphi_3 =& \globally_{[0,20.0]}((x - 5)^2 + (y-5)^2 > 2.0) \;\; \wedge \\
               & \eventually_{[0,20.0]}(x-10)^2 + y^2 < 2.0) \;\; \wedge \\ 
               & \eventually_{[0,10.0]}(x^2 + (y-10)^2 < 2.0).
\end{align*}

\vspace{-1mm}
\noindent
Our algorithm computed a solution in less than $5$ seconds. Fig.~\ref{fig:nonlinearsystem} shows the resulting trajectory, with $((x - 5)^2 + (y-5)^2 > 2.0)$, $((x-10)^2 + y^2 ) < 2.0$ and $(x^2 + (y-10)^2 < 2.0)$ shown as red, blue and green regions respectively. 
\vspace{-1mm}
\begin{figure}[h]
    \centering
    \includegraphics[width=0.37\linewidth]{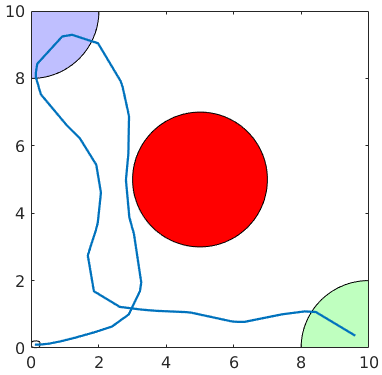}
    \vspace{-1mm}
    \caption{\small Computed trajectory for $\varphi_3$ in blue. The system starts at the bottom left, and the controller drives it to the blue region within 10 seconds, and then to the green region within 20 seconds, while avoiding the red region.}
    \label{fig:nonlinearsystem}
\end{figure}
\section{Conclusion and Future Work}

This paper proposes an algorithmic framework for automaton-guided control synthesis under STL specifications, for both nonlinear dynamics and predicates. The proposed framework is correct by construction and probabilistically complete, and empirical evaluation demonstrates its effectiveness and efficiency. For future work, a natural extension is to generalize our approach to nested STL specifications. Additionally, we plan to consider dynamics and measurement uncertainty by incorporating tree-search approaches such as in \cite{ho2022gbt}, which would allow such control systems to be more readily applied in real world scenarios.





\bibliographystyle{IEEEtran}
\bibliography{bibliography}

\end{document}